\documentclass[twocolumn,showpacs,amsmath,amssymb,pra,nofootinbib]{revtex4}
\usepackage{theorem}
\usepackage{graphicx}
\usepackage{color}
\newcommand{\textchange}[1]{#1}
\theorembodyfont{\rmfamily}   

\newtheorem{lemma}{Lemma}

\newtheorem{proposition}{Proposition}
\newtheorem{remark}{Remark}
\newcommand{\qed}{\hfill$\square$}

\newenvironment{proof}{%
  \noindent{\em Proof.\ }}{%
  \hspace*{\fill}\qed \\
  \vspace{2ex}}
\DeclareMathAlphabet{\bm}{OML}{cmm}{b}{it}
\newcommand{\ket}[1]{| #1 \rangle} 
\newcommand{\bra}[1]{\langle #1 |}

\newcommand{\rom}[1]{\mathrm{#1}}
\newcommand{\san}[1]{\mathsf{#1}}

\newcommand{\argmax}{\mathop{\rm argmax}\limits}

\begin{document}

\title{Tomography increases key rates of quantum-key-distribution protocols}
\thanks{To be published in Physical Review A.}
\author{Shun Watanabe}
 \email{shun-wata@it.ss.titech.ac.jp}
\author{Ryutaroh Matsumoto}%
 \email{ryutaroh@rmatsumoto.org}
 \homepage{http://www.rmatsumoto.org/research.html}
\author{Tomohiko Uyematsu}
 \email{uyematsu@ieee.org}
 \affiliation{%
Department of Communications and Integrated Systems \\
Tokyo Institute of Technology \\
2-12-1, Oookayama, Meguro-ku, Tokyo, 152-8552, Japan
}

\received{18 February 2008}
\revised{16 July 2008}
\accepted{18 July 2008}

\begin{abstract}
We construct a practically implementable classical processing
for the BB84 protocol and the six-state protocol
that fully utilizes the accurate channel estimation method,
which is also known as the quantum tomography.
Our proposed processing yields at least   
as high key rate as the standard processing by Shor and Preskill.  
We show two examples of quantum channels over which
the key rate of our proposed processing is strictly 
higher than the standard processing. 
In the second example, the BB84 protocol with our proposed processing
yields a positive key 
rate even though the so-called error rate is higher than
the 25\% limit.
\end{abstract}

\maketitle

\section{Introduction}         

Quantum key distribution (QKD) has attracted great attention
as an unconditionally secure key distribution scheme.
The fundamental feature of QKD protocols is that
the amount of information gained by an eavesdropper,
usually referred to as Eve, can be estimated
from the channel between the legitimate sender and the receiver,
usually referred to as Alice and Bob respectively.
Such a task cannot be conducted in classical key
distribution schemes.
If the estimated amount is lower than a threshold,
then Alice and Bob determine the length of a secret key
from the estimated amount of Eve's information,
and can share a secret key by performing the 
information reconciliation (error correction)
\cite{brassard:94, bennett:88} 
and the privacy amplification \cite{bennett:88, bennett:95}.
Since the key rate, which is the length of securely sharable key per channel
use, is one of the most important criteria for the efficiency of QKD protocols,
the estimation of the channel is of primary importance.

In this paper, we only treat the BB84 protocol \cite{bennett:84} and
the six-state protocol \cite{bruss:98}, and
we mean the BB84 protocol and the six-state protocol 
by the QKD protocols throughout the paper.
Furthermore, a classical processing consists of 
a procedure to  determine 
a key rate from a channel estimate and a procedure for 
the information reconciliation and the privacy amplification. 

Mathematically, quantum channels are described by 
trace preserving completely positive (TPCP) maps
\cite{nielsen-chuang:00}.
Conventionally in the QKD protocols, we only use 
the statistics of matched measurement outcomes,
which are transmitted and received by the same basis,
to estimate the TPCP map describing the quantum channel;
mismatched measurement outcomes, which are
transmitted and received by different bases, are discarded
in the conventionally used channel estimation methods.
By using the statistics of mismatched measurement outcomes
in addition to that of matched measurement outcomes, we can estimate 
the TPCP map more accurately than the conventional 
estimation method. Such an accurate channel estimation method 
is also known as the quantum tomography \cite{chuang:97, poyatos:97}.
In early 90s, Barnett et al.~\cite{barnett:93} 
showed that the use of mismatched measurement
outcomes enables Alice and Bob to detect the presence of Eve with 
higher probability for the so-called intercept and resend attack.
Furthermore, some literatures use the accurate estimation method to ensure
the channel to be a Pauli channel
\cite{bruss:03,liang:03, kaszlikowski:05, kaszlikowski:05b},
where a Pauli channel is a channel over which four kinds of 
Pauli errors (including the identity) occur probabilistically.
However the channel is not necessarily a Pauli channel.

The use of the accurate channel estimation method 
in a classical processing has a potential 
to improve the key rates of previously known classical processing.
However, there is no proposed practically implementable classical processing that 
fully utilizes the accurate estimation method.
Recently, Renner et al.~\cite{renner:05, renner:05b, kraus:05}
developed information theoretical techniques to
prove the security of the QKD protocols. 
Their proof techniques can be used to prove the security
of the QKD protocols with a classical processing that fully utilizes
the accurate estimation method.
However they only considered Pauli channels or
partial twirled  
channels\footnote{By the partial twirling (discrete twirling)
\cite{bennett:96b}, any channel becomes a Pauli channel.}.
For Pauli channels, the accurate estimation method
and the conventional estimation method make no difference.

In this paper,
we construct a practically implementable 
classical processing that fully utilizes the accurate channel estimation method.
More precisely, we present a procedure to determine a key rate
based on the accurate channel estimate
for the BB84 protocol 
and the six-state protocol respectively.
Then we also present a practically implementable procedure for the information
reconciliation and the privacy amplification 
in which we can share a secret key at the determined key rate.
Note that we only change the classical processing of the QKD protocols,
and the method of the transmission and reception of quantum systems in the
QKD protocols
remain unchanged. 

Although it is straight forward to determine a key rate from
the accurate channel estimate for the six-state protocol,
it is subtle how to determine a key rate from 
the accurate channel estimate for the BB84 protocol.
More specifically, we can obtain
only partial parameters describing the channel,
and there remain some free parameters. Thus we have
to consider the worst case, i.e., the key rate that
is minimized over the free parameters. We shall show an
explicit procedure to determine the minimized key rate.

Our proposed processing yields at least as high key rate
as the standard processing by Shor and Preskill \cite{shor:00}.
As examples, we show that the key rate of our proposed classical processing
is strictly higher than that of the standard processing for the amplitude 
damping channel and the rotation channel, which are unitary 
channel that rotate 
the Bloch sphere in the $\san{z}$-$\san{x}$ plane.
In the example of the amplitude damping channel,
we show that the key rate of the so-called reverse 
reconciliation\footnote{The reverse reconciliation was originally 
proposed by Maurer \cite{maurer:93} in the classical key agreement.},
in which the key is generated based on Bob's bit sequence,
is higher than the key rate of the direct reconciliation,
in which the key is generated based on Alice's bit 
sequence\footnote{For QKD protocols with weak coherent states,
literatures \cite{boileau:05, hayashi:07} already pointed out 
that the key rate of the direct
reconciliation and the reverse reconciliation are different.}. 
In the example of the rotation channel, we solve a
problem left open in \cite[Section 5]{matsumoto:07}---the
problem whether it is possible to obtain positive key rates
from both matched measurement outcomes and 
mismatched measurement outcomes for the BB84 protocol. 

It is believed that we cannot share any secret key if
the so-called error rate is higher than the $25$\% limit 
in the BB84 protocol \cite{gottesman:03}.
However Curty et al.~\cite{curty:04} suggested that,
for some asymmetric error patterns, it might be possible to
share a secret key even for the error
rates above the $25$\% limit.
In the example of the rotation channel, we show 
that we can actually obtain a positive key rate even
though the error rate is higher than the $25$\% limit.

Devetak and Winter \cite{devetak:04} also showed the key rate formula
that coincide with the key rate formula shown 
by Renner et al.~\cite{renner:05, renner:05b, kraus:05} if 
we know the channel exactly.
By combining our proposed procedure to determine a key rate
based on the accurate channel estimate and
Devetak and Winter's procedure for the 
information reconciliation and the privacy amplification,
we can obtain the same key rate as in this paper.
However the procedure for the information reconciliation
and the privacy amplification
shown by Devetak and Winter is not
practically implementable.

Our proposed information reconciliation can be
implemented by any efficiently decodeable linear code
for the Slepian-Wolf coding \cite{slepian:73}.
For example, we can use the low density parity check matrix (LDPC)
code \cite{gallager:63}.

The rest of this paper is organized as follows:
We first present a procedure for the information reconciliation and
the privacy amplification in Section \ref{section-protocol}.
Then we present a procedure to 
determine a key rate from the estimate of the channel
in Section \ref{section-analysis-of-key-rate}.
\textchange{We consider 
the amplitude damping channel, the unital channel, 
and the rotation channel as
examples,
and show that the key rate of our proposed processing
is higher than the standard processing
in Section \ref{section-example}.} 
We state the conclusion in Section \ref{section-conclusion}.

In this paper, we mainly consider 
standard procedures for the information reconciliation
and the privacy amplification with
one-way classical communication, i.e., we 
do not treat, except in Remarks \ref{remark-degradable} 
and \ref{remark-degraded}, 
the noisy preprocessing \cite{renner:05, kraus:05}  
nor a procedure 
with two-way classical communication \cite{gottesman:03, watanabe:07}.
However, our results in this paper can be  easily extended
to procedures with the noisy preprocessing
and two-way classical communication 
(see Remark \ref{remark-extention}).

\section{Information reconciliation and privacy amplification}
\label{section-protocol}

We construct practical procedure for the information
reconciliation and the privacy amplification  in this section.
We first describe our proposed procedure with general linear codes and
the maximum a posteriori
probability (MAP) decoding. 
Then as an example of efficiently decodeable linear code, we show how to apply  
the sum-product algorithm 
of the low density parity check matrix (LDPC) 
code\footnote{It should be noted that the application of the LDPC codes
for classical key agreement protocols has been considered
by Muramatsu \cite{muramatsu:06b}, in which he uses the LDPC code
as the Slepian-Wolf source coding.} 
to our proposed procedure 
in Remark \ref{remark-ldpc}. 

For the simplicity we assume that Eve's attack is  the collective 
attack\footnote{This assumption is not essential.
By using the de Finetti representation
arguments \cite{renner:05b, renner:07}, our result can be extended
to the coherent attack.}, i.e., the channel connecting
Alice and Bob is given by tensor products of a channel $\mathcal{E}_B$
from a qubit density matrix to itself.
As is usual in QKD literatures, we assume that Eve can access 
all the environment of channel $\mathcal{E}_B$;
the channel to the environment is denoted by $\mathcal{E}_E$. 

In the six-state protocol,
Alice randomly sends bit $0$ or $1$ to Bob by modulating it
into a transmission basis that is randomly chosen from
the $\san{z}$-basis $\{ \ket{0_\san{z}}, \ket{1_\san{z}} \}$,
the $\san{x}$-basis $\{ \ket{0_\san{x}}, \ket{1_\san{x}} \}$, 
or the $\san{y}$-basis $\{ \ket{0_\san{y}}, \ket{1_\san{y}} \}$,
where $\ket{0_\san{a}}, \ket{1_\san{a}} $ are eigenstates of the Pauli 
matrix $\sigma_\san{a}$ for $\san{a} \in \{\san{x},\san{y},\san{z}\}$ respectively.  
Then Bob randomly chooses one of measurement observables
$\sigma_\san{x}$, $\sigma_\san{y}$, and $\sigma_\san{z}$, and converts
a measurement result $+1$ or $-1$ into
a bit $0$ or $1$ respectively.
After a sufficient number of transmissions, Alice
and Bob publicly announce their transmission bases and
measurement observables. 
They also announce a part of their bit sequences for
estimating channel $\mathcal{E}_B$.
Note that Alice and Bob do not discard mismatched measurement outcomes,
which are transmitted and received by different bases,
to estimate the channel accurately.

In the BB84 protocol, Alice only uses $\san{z}$-basis and
$\san{x}$-basis to transmit the bit sequence, 
and Bob only uses observable $\sigma_\san{z}$ and 
$\sigma_\san{x}$ to receive the bit sequence. 

Henceforth, we only treat Alice's bit sequence
$\mathbf{x} \in \mathbb{F}_2^n$ that is transmitted in $\san{z}$-basis
and corresponding Bob's bit sequence $\mathbf{y} \in \mathbb{F}_2^n$
that is received in $\sigma_\san{z}$-measurement, where $\mathbb{F}_2$ is
the finite field of order $2$.
Furthermore, we occasionally omit the subscripts $\{ \san{x}, \san{y}, 
\san{z} \}$ 
of bases, and the basis $\{ \ket{0}, \ket{1} \}$ is regarded as
$\san{z}$-basis unless otherwise stated.
Since the pair of sequences $(\mathbf{x}, \mathbf{y})$ is
transmitted and received in $\san{z}$-basis, they are independently identically distributed 
according to 
\begin{eqnarray}
\label{eq-distribution-xy}
P_{XY}(x,y) := \frac{1}{2} \bra{y_\san{z}} \mathcal{E}_B(
	\ket{x_\san{z}}\bra{x_\san{z}}) \ket{y_\san{z}}.
\end{eqnarray} 
Note that the distribution $P_{XY}$ can be estimated 
from the statistics of the sample bits that are transmitted by $\san{z}$-basis
and received by $\sigma_\san{z}$-observable.
 
Before describing our proposed procedure, we should review the basic
facts of linear codes.
An $[n,n-m]$ classical linear code $\mathcal{C}$ is an 
$(n-m)$-dimensional linear subspace of $\mathbb{F}_2^n$, and its
parity check matrix $M$ is an $m \times n$ matrix
of rank $m$ with $0,1$ entries such that $M \mathbf{c} = \mathbf{0}$
for any codeword $\mathbf{c} \in \mathcal{C}$.
By using these preparations, our proposed procedure is described as follows.
\begin{enumerate}
\renewcommand{\theenumi}{\roman{enumi}}
\renewcommand{\labelenumi}{(\theenumi)}

\item \label{step1}
Alice calculates syndrome $\mathbf{t} := M \mathbf{x}$,
and sends it to Bob over the public channel.

\item \label{step2}
Bob decodes $(\mathbf{y}, \mathbf{t})$ into estimate $\hat{\mathbf{x}}$ of
$\mathbf{x}$ by using the maximum a posteriori probability (MAP) decoding.
More precisely, Bob selects $\hat{\mathbf{x}} \in \mathbb{F}_2^n$ such that
$M \hat{\mathbf{x}} = \mathbf{t}$ and
a posteriori probability $P_{X|Y}^n(\hat{\mathbf{x}}|\mathbf{y})$ is
maximized (if there exist tied sequences, then he selects 
the smallest one with respect to the lexicographic order),
where $P_{X|Y}^n$ is the $n$th product distribution of $P_{X|Y}$.

\item \label{step3}
Alice randomly choose a hash function $f: \mathbb{F}_2^n \to \mathcal{S}_n$
from a set of universal hash functions \cite{carter:79},
and sends the choice to Bob over the public channel. Then Alice
and Bob's final keys are 
$s_A := f(\mathbf{x})$ and $s_B := f(\hat{\mathbf{x}})$ respectively.
\end{enumerate}

If we set the rate of syndrome  as 
\begin{eqnarray}
\label{rate-of-syndrome}
\frac{m}{n} > H(X|Y),
\end{eqnarray}
then there exists a linear code in the LDPC codes
such that Bob's decoding error probability is arbitrary small for sufficiently
large $n$ \cite[Theorem 2]{muramatsu:05}, 
where $H(X|Y)$ is the conditional entropy with respect to
the joint probability distribution $P_{XY}$ \cite{cover}.
Note that the base of a logarithm and a (conditional) entropy
are $2$ throughout the paper.

The key rate, $\frac{1}{n} \log |\mathcal{S}_n|$, is determined according
to the results 
of privacy amplification \cite[Corollary 3.3.7 and Lemma
6.4.1]{renner:05b}.
Let 
\begin{eqnarray*}
H_{\rho}(X|E) := H(\rho_{XE}) - H(\rho_E) 
\end{eqnarray*}
be the conditional von Neumann entropy 
with respect to density matrix
$\rho_{XE} := \sum_{x \in \mathbb{F}_2} \frac{1}{2} \ket{x}\bra{x}
\otimes \mathcal{E}_E(\ket{x}\bra{x})$,
where $H(\rho)$ is the von Neumann entropy for a density matrix $\rho$.
If the key rate satisfies
\begin{eqnarray}
	\label{eq-key-rate}
\frac{1}{n} \log |\mathcal{S}_n| <  H_{\rho}(X|E) - \frac{m}{n},
\end{eqnarray}
then the final key $S_A$ is secure in the sense of the trace 
distance\footnote{The trace norm of a matrix $A$ is defined by $\| A\| :=
\rom{Tr} \sqrt{A^* A}$. Then the trace distance between two matrices
$A$ and $B$ is defined by $\| A - B\|$.}.
More precisely, the density matrix, $\rho_{S_A \mathbf{T} F E^n}$, which
describes Alice's final key $S_A$, 
the publicly transmitted syndrome $\mathbf{T}$ and hash function $F$, 
and the state in Eve's system $E^n$,
satisfies
\begin{eqnarray*}
\| \rho_{S_A \mathbf{T} F E^n} - \rho_S \otimes \rho_{\mathbf{T} F E^n} \| \le \varepsilon
\end{eqnarray*}
for arbitrary small $\varepsilon$ and sufficiently large $n$, where 
$\rho_S := \sum_{s \in \mathcal{S}_n} \frac{1}{|\mathcal{S}_n|} \ket{s}\bra{s}$ is
the density matrix that describes the uniformly distributed key on
$\mathcal{S}_n$.
From Eqs.~(\ref{rate-of-syndrome}) and (\ref{eq-key-rate}), we find that
\begin{eqnarray}
\label{eq-secure-key-rate}
H_\rho(X|E) - H(X|Y)
\end{eqnarray}
is a secure key rate.

Note that the conditional von Neumann entropy $H_\rho(X|E)$
can be calculated from the channel $\mathcal{E}_B$ as follows.
Since system $X$ is classical, we can rewrite 
$H(\rho_{XE}) = H(X) + \sum_{x \in \mathbb{F}_2} \frac{1}{2} 
H(\mathcal{E}_E(\ket{x}\bra{x}))$.
Noting that  $H(\mathcal{E}_E(\ket{x}\bra{x})) = H(\mathcal{E}_B(\ket{x}\bra{x}))$
and $H(\rho_E) = H((\rom{id} \otimes \mathcal{E}_B)(\psi))$ for
the maximally entangled state 
$\ket{\psi} := \sum_{x \in \mathbb{F}_2} \frac{1}{\sqrt{2}} \ket{x}\ket{x}$,
Eve's ambiguity for Alice's bit, $H_{\rho}(X|E)$, can be calculated from the channel $\mathcal{E}_B$.
How to determine Eve's ambiguity $H_\rho(X|E)$ from a estimate of the
channel $\mathcal{E}_B$
is discussed in the next section.  

\begin{remark}
\label{remark-asymmetric}
If we use the conventionally used method \cite{mayers:01, shor:00}
for decoding $\hat{\mathbf{x}}$, the rate of syndrome $\frac{m}{n}$
cannot be as small as the right hand side of Eq.~(\ref{rate-of-syndrome}).
Thus, the key rate in Eq.~(\ref{eq-secure-key-rate}) cannot be achieved.
Define a probability distribution on $\mathbb{F}_2$ as
\begin{eqnarray}
\label{eq-pw}
P_W(w) := \sum_{y \in \mathbb{F}_2} P_Y(y) P_{X|Y}(y+w|y).
\end{eqnarray}
Then the error $\mathbf{w} := \mathbf{x} + \mathbf{y}$
between Alice and Bob's sequence is distributed according to $P_W^n$. 
In the conventional method, 
Bob calculates the difference of syndromes,
$\mathbf{t} + M \mathbf{y}$, and selects the error
$\hat{\mathbf{w}}$ such that $M \hat{\mathbf{w}} = \mathbf{t} + M
 \mathbf{y}$
and the likelihood of the error 
$P_{W}^n(\hat{\mathbf{w}})$ is maximized.
Then , the estimate for Alice's sequence is 
$\hat{\mathbf{x}} = \mathbf{y} + \hat{\mathbf{w}}$.
The rate of syndrome have to be larger than $H(W)$ 
for the decoding error probability to be small.
By the log-sum inequality \cite{cover} and Eq.~(\ref{eq-pw}), we have
\begin{eqnarray*}
\lefteqn{ H(X|Y) } \\
&=& \sum_{x, y \in \mathbb{F}_2}  P_Y(y) 
    P_{X|Y}(x|y) \log \frac{1}{P_{X|Y}(x|y)} \\
&=& \sum_{w ,y \in \mathbb{F}_2} 
    P_{Y}(y) P_{X|Y}(y+w|y) \log 
       \frac{P_{Y}(y)}{P_{Y}(y) P_{X|Y}(y+w|y)} \\
&\le& \sum_{w \in \mathbb{F}_2} P_W(w) \log
   \frac{1}{P_W(w)} \\
&=& H(W).
\end{eqnarray*}
Thus, the key rate in Eq.~(\ref{eq-secure-key-rate}) cannot be achieved by
the conventional decoding method unless
$P_{X|Y}(w|0)$ equals $ P_{X|Y}(1 + w|1)$ for any $w \in \mathbb{F}_2$.
\end{remark}
\begin{remark}
By switching the role of Alice and Bob, we obtain 
a classical processing that achieves the key rate
\begin{eqnarray}
\label{key-rate-reverse}
H_{\rho}(Y|E) - H(Y|X).
\end{eqnarray}
Such a procedure is usually called the reverse
reconciliation. On the other hand the original procedure is usually called
the direct reconciliation. 
The reverse reconciliation was originally 
proposed by Maurer in the classical key agreement 
context \cite{maurer:93}.

Note that we can calculate
the conditional von Neumann entropy $H_\rho(Y|E) = H(\rho_{YE}) - H(\rho_E)$
from the channel $\mathcal{E}_B$ as follows.
Let $\psi_{ABE}$ be a purification of 
$(\rom{id} \otimes \mathcal{E}_B)(\psi)$,
and let $\rho_{BE} := \rom{Tr}_A[ \psi_{ABE}]$. 
Then, the density matrix $\rho_{YE}$ is derived by measurement on Bob's system,
i.e., 
\begin{eqnarray*}
\rho_{YE} = 
\sum_{y \in \mathbb{F}_2} (\ket{y}\bra{y} \otimes I) \rho_{BE}
(\ket{y}\bra{y} \otimes I).
\end{eqnarray*}

In Section \ref{example1}, we shall show that
the key rate of the reverse reconciliation can be higher than that of 
the direct reconciliation. 
The fact that the key rate of the direct reconciliation
and the reverse reconciliation are different is already pointed
out for QKD protocols with weak coherent states \cite{boileau:05, hayashi:07}. 
\end{remark}
\begin{remark}
We used the MAP decoding instead of
the maximum likelihood (ML) decoding in our procedure,
because the MAP decoding minimizes the decoding error probability, and 
the MAP decoding is different from the ML decoding 
for the reverse reconciliation.
In the ML decoding for the reverse reconciliation, 
Alice selects $\hat{\mathbf{y}} \in \mathbb{F}_2^n$
such that $M \hat{\mathbf{y}}$ equals the syndrome 
$\mathbf{t} = M \mathbf{y}$,  
and that the likelihood
$P_{X|Y}^n(\mathbf{x}|\hat{\mathbf{y}})$ is maximized.
Since the prior probability of Bob's sequence $\mathbf{y}$ 
is not necessarily
the uniform distribution, the ML decoding and the MAP
decoding are not necessarily equivalent, i.e.,
\begin{eqnarray*}
\argmax_{\hat{\mathbf{y}}: M \hat{\mathbf{y}} = \mathbf{t}}
 P_{X|Y}^n(\mathbf{x}|\hat{\mathbf{y}}) =
 \argmax_{\hat{\mathbf{y}}: M \hat{\mathbf{y}} = \mathbf{t}}
 P_{Y|X}^n(\hat{\mathbf{y}}|\mathbf{x})
\end{eqnarray*}
does not hold in general.
\end{remark}
\begin{remark}
\label{remark-unmatched}
By modifying our proposed procedure as follows,
we obtain a procedure in which Alice and Bob can
share a secret key from Alice's sequence $\mathbf{x}$ that is transmitted by $\san{z}$-basis
and corresponding Bob's sequence $\mathbf{y}$ that is received by
$\sigma_\san{x}$-measurement.
Since $(\mathbf{x}, \mathbf{y})$ are independently identically
 distributed according to 
\begin{eqnarray}
\label{eq-distribution-xy-2}
P_{X Y^\prime}(x,y) := \frac{1}{2} \bra{y_\san{x}}
 \mathcal{E}_B(\ket{x_\san{z}}\bra{x_\san{z}}) \ket{y_\san{x}},
\end{eqnarray}
we replace  $P_{X|Y}^n$ in Step (\ref{step2}) with $P_{X|Y^\prime}^n$.
By a similar arguments as in the original procedure,
the secure key rate of the modified procedure is given by
\begin{eqnarray}
\label{key-rate-lower-reverse}
H_\rho(X|E) - H(X|Y^\prime).
\end{eqnarray} 

In Section \ref{example2}, we shall show an example
in which Alice and Bob can share secret keys both from
matched measurement outcomes and mismatched measurement outcomes,
i.e., both Eqs.~(\ref{eq-secure-key-rate}) and
 (\ref{key-rate-lower-reverse}) 
are positive.
\end{remark}
\begin{remark}
\label{remark-ldpc}
The sum product algorithm can be used in
Step (\ref{step2}) of our proposed 
procedure as follows.
For a given sequence $\mathbf{y} \in \mathbb{F}_2^n$,
and a syndrome $\mathbf{t} \in \mathbb{F}_2^m$, define
a function
\begin{eqnarray}
\label{eq-sum-product}
P^*(\hat{\mathbf{x}}) := \prod_{j=1}^n P_{X|Y}(\hat{x}_j|y_j)
 \prod_{k=1}^m \mathbf{1}\left[ \sum_{\ell \in N(k)} 
 \hat{x}_\ell = t_k \right],
\end{eqnarray}
where $N(k) := \{ j \mid M_{kj} = 1\}$ for 
the parity check matrix $M$, and $\mathbf{1}[\cdot]$ is the
indicator function.
The function $P^*(\hat{\mathbf{x}})$ is the non-normalized
a posteriori probability distribution on $\mathbb{F}_2^n$
given $\mathbf{y}$ and $\mathbf{t}$. The sum-product
algorithm is a method to (approximately) calculate
the marginal a posteriori probability, i.e.,
\begin{eqnarray*}
P^*_j(\hat{x}_j) := \sum_{\hat{x}_\ell,\ell \neq j}
  P^*(\hat{\mathbf{x}}).
\end{eqnarray*}
The definition of a posteriori probability in Eq.~(\ref{eq-sum-product})
is the only difference between the decoding for the Slepian-Wolf source coding
and that for the channel coding.
More precisely, we replace \cite[Eq.~(47.6)]{mackay-book} with
Eq.~(\ref{eq-sum-product}), and use the algorithm in 
\cite[Section 47.3]{mackay-book}.
The above procedure is a generalization of \cite{liveris:02},
and a special case of \cite{coleman:06}.

In QKD protocols we should 
minimize the block error probability 
rather than the bit error probability, because a bit error might
propagate to other bits after the privacy amplification.
Although the sum-product algorithm is designed to minimize
the bit error probability, it is known by computer simulations
that the algorithm makes the block error 
probability small \cite{mackay-book}. 
\end{remark}

\section{Procedure for channel estimation}
\label{section-analysis-of-key-rate}

In this section we show procedures to 
estimate Eve's ambiguity $H_{\rho}(X|E)$ for 
the six-state protocol and the BB84 protocol.
We first present general preliminaries 
in Section \ref{preliminaries}.
Then we show procedures for the six-state protocol and 
the BB84 protocol in Sections \ref{six-state} and 
\ref{bb84} respectively.
In Section \ref{subsection-relation-to-partial-twirled-channel},
we clarify the relation between our proposed procedures
for estimating $H_{\rho}(X|E)$ and the conventional
ones. 

Although we explain the procedures to estimate
$H_{\rho}(X|E)$ for the direct reconciliation,
the estimation of $H_{\rho}(Y|E)$ for the reverse
reconciliation can be done in a similar manner.

\subsection{Preliminaries}
\label{preliminaries}

\textchange{ In the Stokes parameterization, the qubit channel $\mathcal{E}_B$
can be described by the affine map parameterized
by $12$ real parameters \cite{fujiwara:98,fujiwara:99}: }
\begin{eqnarray}
\left[ \begin{array}{c}
\theta_{\san{z}} \\ \theta_{\san{x}} \\ \theta_{\san{y}}
\end{array} \right] 
\mapsto
\left[ \begin{array}{ccc}
R_{\san{zz}} & R_{\san{zx}} & R_{\san{zy}} \\
R_{\san{xz}} & R_{\san{xx}} & R_{\san{xy}} \\
R_{\san{yz}} & R_{\san{yx}} & R_{\san{yy}}
\end{array} \right]
\left[ \begin{array}{c}
\theta_{\san{z}} \\ \theta_{\san{x}} \\ \theta_{\san{y}}
\end{array} \right]
+ \left[ \begin{array}{c}
t_{\san{z}} \\ t_{\san{x}} \\ t_{\san{y}} 
\end{array} \right],
\label{eq-affine-map}
\end{eqnarray}
\textchange{where $(\theta_\san{z}, \theta_\san{x}, \theta_\san{y})$
describes a vector in the Bloch sphere \cite{nielsen-chuang:00}}.
\textchange{ For the channel $\mathcal{E}_B$ and each pair of bases
$(\san{a}, \san{b}) \in \{ \san{z}, \san{x}, \san{y} \}^2$,  
define the biases of the outputs as }
\begin{eqnarray*}
Q_{\san{ab}0} &:=& 
\bra{0_{\san{b}}} \mathcal{E}_B( \ket{0_{\san{a}}} \bra{0_{\san{a}}}) \ket{0_{\san{b}}}
- \bra{1_{\san{b}}} \mathcal{E}_B( \ket{0_{\san{a}}} \bra{0_{\san{a}}}) \ket{1_{\san{b}}}, \\
Q_{\san{ab}1} &:=&
\bra{1_{\san{b}}} \mathcal{E}_B( \ket{1_{\san{a}}} \bra{1_{\san{a}}}) \ket{1_{\san{b}}}
- \bra{0_{\san{b}}} \mathcal{E}_B( \ket{1_{\san{a}}} \bra{1_{\san{a}}}) \ket{0_{\san{b}}}.
\end{eqnarray*}
\textchange{ Then, a straight forward calculation shows the relations}
\begin{eqnarray}
\label{eq-relation-between-bias-parameter}
R_{\san{ba}} = \frac{1}{2}(Q_{\san{ab}0} + Q_{\san{ab}1}),~~
t_{\san{b}} = \frac{1}{2}(Q_{\san{ab}0} - Q_{\san{ab}1}).
\end{eqnarray}

\textchange{ The qubit channel $\mathcal{E}_B$ can be also described by
the Choi matrix $\rho_{AB} := (\rom{id} \otimes \mathcal{E}_B)(\psi)$
\cite{choi:75} for the maximally entangled state 
$\ket{\psi} = \frac{1}{\sqrt{2}}(\ket{0}\ket{0} + \ket{1}\ket{1})$.
By using the parameters in Eq.~(\ref{eq-affine-map}),
we can write the Choi matrix $\rho_{AB}$ as }
\begin{widetext}
\begin{eqnarray}
\label{eq-choi-matrix}
\frac{1}{4} \left[ \begin{array}{cccc}
1+R_{\san{zz}} + t_\san{z} 
 & R_{\san{xz}} + t_\san{x} + \mathbf{i}R_{\san{yz}} + \mathbf{i}t_{\san{y}}
 & R_{\san{zx}} - \mathbf{i} R_{\san{zy}} 
 & R_{\san{xx}} + R_{\san{yy}} + \mathbf{i} R_{\san{yx}} - \mathbf{i} R_{\san{xy}} \\
R_\san{xz} + t_\san{x} - \mathbf{i}R_\san{yz} - \mathbf{i} t_\san{y}
 & 1 - R_\san{zz} - t_\san{z}
 & R_\san{xx} - R_\san{yy} - \mathbf{i}R_\san{yx} - \mathbf{i} R_\san{xy}
 & - R_\san{zx} + \mathbf{i} R_\san{zy}    \\
R_\san{zx} + \mathbf{i} R_\san{zy} 
 & R_\san{xx} - R_\san{yy} + \mathbf{i}R_\san{yx} + \mathbf{i} R_\san{xy}
 & 1 - R_\san{zz} + t_\san{z}
 & -R_\san{xy} + t_\san{x}  - \mathbf{i}R_\san{yz} + \mathbf{i} t_\san{y} \\
R_\san{xx} + R_\san{yy} - \mathbf{i}R_\san{yx} + \mathbf{i} R_\san{xy}
 & - R_\san{zx} - \mathbf{i} R_\san{zy}
 & -R_\san{xz} + t_\san{x}  + \mathbf{i} R_\san{yz} - \mathbf{i} t_\san{y}
 & 1 + R_\san{zz} - t_\san{z}    
\end{array} \right],
\end{eqnarray}
\end{widetext}
\textchange{ where $\mathbf{i}$ is the imaginary unit.}

\subsection{Six-state protocol}
\label{six-state}

An ad-hoc approach to estimate
Eve's ambiguity in the six-state protocol is
very simple, because all parameters can be estimated
from the statistics of sampled bits \cite{chuang:97, poyatos:97}. 
\begin{enumerate}
\renewcommand{\theenumi}{\roman{enumi}}
\renewcommand{\labelenumi}{(\theenumi)}

\item \label{estimate-six-state-step1}
\textchange{ By using the statistics of sampled bits and the relation
in Eq.~(\ref{eq-relation-between-bias-parameter}),
Alice and Bob calculate the estimate
$(\tilde{R}, \tilde{t})$ for the parameters of
the channel $\mathcal{E}_B$. }

\item \label{estimate-six-state-step2}
\textchange{ By using Eq.~(\ref{eq-choi-matrix}), Alice and Bob calculate
the corresponding matrix $\tilde{\rho}_{AB}$. } If 
the resulting matrix $\tilde{\rho}_{AB}$ is not a Choi matrix, 
Alice and Bob select a Choi matrix $\hat{\rho}_{AB}$ such that
the Frobenius norm between $\hat{\rho}_{AB}$ and
$\tilde{\rho}_{AB}$ is minimized\footnote{This step can be implemented,
for example, by the convex optimization \cite{boyd-book:04} because
the set of all Choi matrices is a closed convex set.
For more detail, see Appendix \ref{appendix-convex-opt}.}.

\item \label{estimate-six-state-step3}
Alice and Bob calculate an estimator $H_{\hat{\rho}}(X|E)$ for 
Eve's ambiguity $H_\rho(X|E)$.
\end{enumerate}

The validity of this estimation procedure is
shown as follows. Since the estimators in
Step (\ref{estimate-six-state-step1}) converge
to the true parameters in probability
as the number of sampled bits goes to the infinity, 
the matrix $\tilde{\rho}_{AB}$
also converges\footnote{When we consider a convergence of a density
matrix, the convergence is with respect to the trace distance.
On the other-hand, when we consider a convergence of parameters,
we use the Euclidean distance. If estimated parameters converges
to the true values, then the resulting matrix also
converges to the true one, because the convergence of the 
Frobenius norm and that of the trace norm are equivalent.}
to $\rho_{AB}$.
Then the Choi matrix $\hat{\rho}_{AB}$ also
converges to the $\rho_{AB}$.  
Since the conditional entropy is a continuous function,
the estimator $H_{\hat{\rho}}(X|E)$ in Step (\ref{estimate-six-state-step3})
also converges to $H_\rho(X|E)$ in probability
as the number of sampled bits goes to the infinity.

\subsection{BB84 protocol}
\label{bb84}

\textchange{ The estimation of $H_{\rho}(X|E)$ in the BB84 protocol is
much more complicated. When Alice and Bob only use
$\san{z}$-basis and $\san{x}$-basis, the statistics of the input
and the output are irrelevant to the parameters
$(R_{\san{zy}}, R_{\san{xy}}, R_{\san{yz}}, R_{\san{yx}}, R_{\san{yy}}, t_{\san{y}})$.
Thus, we can only estimate the parameters
$\omega = (R_{\san{zz}}, R_{\san{zx}}, R_{\san{xz}}, R_{\san{xx}}, t_{\san{z}}, t_{\san{x}})$,
and we have to consider the worst case for 
the parameters $\omega$, i.e.,  }
\begin{eqnarray}
\label{eq-minimization-1}
F(\omega) 
:= \min_{\tau \in \mathcal{P}^\prime(\omega)} H_{\rho_\tau}(X|E),
\end{eqnarray}
\textchange{ where $\mathcal{P}^\prime(\omega)$ is the set of all 
parameters 
$\tau = (R_{\san{zy}}, R_{\san{xy}}, R_{\san{yz}}, R_{\san{yx}}, R_{\san{yy}}, t_{\san{y}})$
such that the parameters $\omega$
and $\tau$ constitute a qubit channel, and 
$\rho_\tau$ is the Choi matrix corresponding to the parameter $\tau$\footnote{It should be
noted that there are some other papers \cite{ziman:05, ziman:06, ziman:08} that 
consider the situation in which we have to estimate a channel 
from partially estimated parameters as above.
However, the methods in these papers cannot be used in our problem.}. }

By using the following proposition, 
which is proved in Appendix \ref{appendix-a}, we can make
the desired function 
$F(\omega)$ 
into a simpler form.
\begin{proposition}
\label{proposition-non-imaginary}
\textchange{ The minimization in Eq.~(\ref{eq-minimization-1}) is achieved when 
the parameters, $R_{\san{zy}}$, $R_{\san{xy}}$, $R_{\san{yz}}$, $R_{\san{yx}}$,
and $t_{\san{y}}$, are $0$. }
\end{proposition}
The number of free parameters has been reduced to $1$ by 
Proposition \ref{proposition-non-imaginary}. 
Thus the problem is rewritten as looking for an estimator of
\begin{eqnarray}
\label{eq-minimization-2}
F(\omega) 
= \min_{R_{\san{yy}} \in \mathcal{P}(\omega)} H_{\rho_{R_{\san{yy}}}}(X|E),
\end{eqnarray}
\textchange{ where $\mathcal{P}(\omega)$ is the set of parameters $R_{\san{yy}}$ such that
the parameters $\omega$ and $R_{\san{yy}}$ constitute a qubit channel
when other parameters are all $0$, and $\rho_{R_\san{yy}}$ is the Choi
matrix corresponding to the parameter $R_\san{yy}$.
Since the range $\mathcal{P}(\omega)$ of the remaining free parameter $R_{\san{yy}}$
is a closed interval and
$H_{\rho}(X|E)$ is a convex function 
(see Lemma \ref{proposition-convexity}),
the minimization in $F(\omega)$ is achieved at the boundary 
point of the range of $R_{\san{yy}}$ or at the zero point 
of the derivative of $H_{\rho}(X|E)$ with respect to $R_{\san{yy}}$. }

An ad-hoc approach to find an estimator is as follows.
\begin{enumerate}
\renewcommand{\theenumi}{\roman{enumi}}
\renewcommand{\labelenumi}{(\theenumi)}
\item 
\textchange{ By using the statistics of sampled bits and the relation
in Eq.~(\ref{eq-relation-between-bias-parameter}),
Alice and Bob calculate the estimate
$\tilde{\omega}$ for the parameters $\omega$. }

\item
If $\mathcal{P}(\tilde{\omega})$ is the empty set, then
Alice and Bob find the point $\hat{\omega}$
such that $\hat{\omega}$ is closest (in Euclidean distance)
to $\tilde{\omega}$ and 
$\mathcal{P}(\hat{\omega})$ is not an empty set\footnote{This
step can be implemented, for example, by the convex optimization
\cite{boyd-book:04} because the set of all $\hat{\omega}$s
such that ${\cal P}(\hat{\omega})$ is not empty is a 
closed convex set. For more detail, see Appendix \ref{appendix-convex-opt}.}.

\item 
Alice and Bob calculate an estimator $F(\hat{\omega})$
for Eve's (worst-case) ambiguity $F(\omega)$.
\end{enumerate}
The validity of this estimation procedure can be
shown as follows. 
The estimator $\tilde{\omega}$ converges
to the true value $\omega$ in probability.
The estimator $\hat{\omega}$ also converges to $\omega$,
because $\| \hat{\omega} - \tilde{\omega} \| \le \| \tilde{\omega} -
\omega \|$, which implies 
$\|\hat{\omega} - \omega \| \le 2 \|\tilde{\omega} - \omega \|$
by the triangle inequality.
Thus the following lemma, which is proved in Appendix \ref{appendix-b},
guarantees that 
the estimator $F(\hat{\omega})$ converges to the desired quantity
$F(\omega)$ in probability
as the number of sampled bits goes to the infinity.
\begin{lemma}
\label{lemma-continuity}
The function $F(\omega)$ is a continuous function
of $\omega$.
\end{lemma}

\textchange{ Although we showed a procedure to exactly estimate Eve's
worst case ambiguity so far, it is worthwhile to show
a closed form lower bound on Eve's worst case ambiguity,
which will be proved in Appendix \ref{appendix-d}. }
\begin{proposition}
\label{remark-unital}
\textchange{ Let $d_{\san{z}}$ and $d_{\san{x}}$ be the singular values of
the matrix }
\begin{eqnarray}
\label{eq-matrix}
\left[ \begin{array}{cc} 
R_{\san{zz}} & R_{\san{zx}} \\
R_{\san{xz}} & R_{\san{xx}}
\end{array} \right].
\end{eqnarray}
\textchange{ Then, we have} 
\begin{eqnarray}
F(\omega) 
&\ge& 1 - h\left( \frac{1 + d_{\san{z}}}{2} \right) 
- h\left( \frac{1 + d_{\san{x}}}{2} \right) \nonumber \\
&&~~~~+ h\left( \frac{1 + \sqrt{R_{\san{zz}}^2 + R_{\san{xz}}^2}}{2} \right),
\label{eq-unital-bound}
\end{eqnarray}
\textchange{ where $h(\cdot)$ is the binary entropy function. 
The equality holds if $t_{\san{z}} = t_{\san{x}} = 0$. }
\end{proposition}
\begin{remark}
\label{remark-after-unital}
\textchange{ For the reverse reconciliation, the worst case of Eve's
ambiguity $H_{\rho}(Y|E)$ is lower bounded by the right hand side
of Eq.~(\ref{eq-unital-bound}) in which $R_{\san{xz}}$ is replaced
by $R_{\san{zx}}$. }
\end{remark}
\begin{remark}
\textchange{ The right hand side of Eq.~(\ref{eq-unital-bound}) is further lower bounded by
$1 - h((1-R_\san{xx})/2)$. Since $(1-R_\san{xx})/2$ equals to the 
so-called phase error rate $P_\san{x}$ (see Eq.~(\ref{eq-phase-error})),
the right hand side of Eq.~(\ref{eq-unital-bound}) is a lower bound
on Eve's worst case ambiguity that is tighter than the well known
bound $1 - h(P_\san{x})$ \cite{renner:05}. }
\end{remark}

\begin{remark}
We described estimation methods for Eve's ambiguity $H_\rho(X|E)$
based on the channel estimation method so-called linear inversion
\cite{hradil:04} in Section \ref{six-state} and in this section.
It is well-known that the maximum likelihood (ML) 
channel estimator
has smaller estimation error than the linear inversion \cite{hradil:04}.
An algorithm for ML channel estimation has been proposed
\cite{fiurasek:01, jezek:03, hradil:04}, 
however, its convergence as a numerical algorithm
has not been proved. The absence of a convergence proof prevents
us from using that algorithm in the QKD protocols
that require a rigorous proof of the convergence of an estimator.

The computation of the ML channel estimate in the six-state
protocol is a convex optimization
problem.
Because the set of Choi matrices is a closed convex set defined by equality
constraints and generalized inequality constraints \cite{boyd-book:04}
and the log-likelihood function is a concave function of Choi matrices
for given measurement outcomes.
Therefore,
the interior point method \cite{boyd-book:04}, for example,
 can compute the ML estimate with
convergence guarantee.
For the BB84 protocol, the domain of log-likelihood function
is narrowed to real Choi matrices by Proposition \ref{proposition-non-imaginary} that
 is also a closed convex set, and the parameter $R_{\san{yy}}$ remains undetermined
as well as the linear inversion
because the log-likelihood function is independent of $R_{\san{yy}}$.
The rest of parameters can be computed by a convex optimization algorithm.
If we are allowed to use enough computation time for sophisticated 
channel estimation procedures,
then it may be better to use the ML channel estimation. 
\end{remark}

\subsection{Relation to the conventional estimation procedure}
\label{subsection-relation-to-partial-twirled-channel}

In this section, we show the relation between Eve's ambiguity
$H_{\rho}(X|E)$ that is estimated by our proposed procedures
and that estimated by the conventional procedures.

\textchange{ In the conventional procedure to estimate $H_{\rho}(X|E)$ in the 
six-state protocol \cite{renner:05}, we first estimate the so called the 
error rate for each basis: }
\begin{eqnarray}
P_\san{z} &:=& \frac{ \bra{1_{\san{z}}} \mathcal{E}_B(\ket{0_{\san{z}}}\bra{0_{\san{z}}})
\ket{1_{\san{z}}} + \bra{0_{\san{z}}} \mathcal{E}_B(\ket{1_{\san{z}}}\bra{1_{\san{z}}})
\ket{0_{\san{z}}} }{2}, \nonumber \\
P_\san{x} &:=& \frac{ \bra{1_{\san{x}}} \mathcal{E}_B(\ket{0_{\san{x}}}\bra{0_{\san{x}}})
\ket{1_{\san{x}}} + \bra{0_{\san{x}}} \mathcal{E}_B(\ket{1_{\san{x}}}\bra{1_{\san{x}}})
\ket{0_{\san{x}}} }{2}, \label{eq-phase-error} \\
P_\san{y} &:=& \frac{ \bra{1_{\san{y}}} \mathcal{E}_B(\ket{0_{\san{y}}}\bra{0_{\san{y}}})
\ket{1_{\san{y}}} + \bra{0_{\san{y}}} \mathcal{E}_B(\ket{1_{\san{y}}}\bra{1_{\san{y}}})
\ket{0_{\san{y}}} }{2}. \nonumber 
\end{eqnarray}
\textchange{ Then, we calculate the worst case of Eve's ambiguity
$\min H_{\rho}(X|E)$ in which the minimization is taken
over the set of all channels that are compatible with 
the estimates of the error rates
$(P_\san{z}, P_\san{x}, P_\san{y})$. Since we estimate the actual channel
instead of the worst case, Eve's ambiguity estimated by our procedure
is at least as large as that estimated by the conventional one. }

\textchange{ In the conventional procedure to estimate $H_\rho(X|E)$ in
the BB84 protocol, we first estimate $P_{\san{z}}$
and $P_\san{x}$. Then we calculate
the worst case of Eve's ambiguity
$\min H_{\rho}(X|E)$ in which the minimization is taken
over the set of all channels that are compatible with 
the estimates of the error rates
$(P_\san{z}, P_\san{x})$. The minimum is given by the
well known value $1 - h(P_\san{x})$
\cite{renner:05}. Since the error rates $(P_\san{z}, P_\san{x})$ are 
degraded version of the parameters $\omega$, the range of
minimization in the conventional procedure is larger than
$\mathcal{P}(\omega)$ in our proposed procedure.
Thus, Eve's worst case ambiguity estimated by our proposed 
procedure is at least as large as that estimated by the
conventional one. }

\textchange{ For both the six-state protocol and the BB84 protocol,
a sufficient condition such that Eve's worst case ambiguity
estimated by our proposed procedure and that estimated by
the conventional one coincide is
that the channel $\mathcal{E}_B$ is a Pauli channel.
However, it is not clear 
whether the condition 
is also a necessary condition or not. }

Combining the arguments in this section and 
Remark \ref{remark-asymmetric}, we find that our
proposed classical processing yields at least as high key rate as
the standard processing by Shor and Preskill \cite{shor:00}
for the QKD protocols.

\section{Examples}
\label{section-example}

In this section, we calculate the key rates of
the BB84 protocol and the six-state protocol with
our proposed classical processing for the amplitude damping 
channel, \textchange{the unital channel, and the rotation channel}, 
and clarify that the key rate of our proposed classical processing is higher 
than previously known ones.

\subsection{Amplitude damping channel}
\label{example1}

\textchange{ In the Stokes parameterization, 
the amplitude damping 
channel $\mathcal{E}_p$ is given by the afine map }
\begin{eqnarray}
\left[ \begin{array}{c}
\theta_{\mathsf{z}} \\ \theta_{\mathsf{x}} \\ \theta_{\mathsf{y}}
\end{array} \right]
\mapsto
\left[ \begin{array}{ccc}
1 - p & 0 & 0 \\
0 & \sqrt{1-p} & 0 \\
0 & 0 & \sqrt{1-p}
\end{array}
\right]
\left[ \begin{array}{c}
\theta_{\mathsf{z}} \\ \theta_{\mathsf{x}} \\ \theta_{\mathsf{y}}
\end{array} \right]
+  
\left[ \begin{array}{c}
p \\ 0 \\ 0
\end{array} \right]
\label{eq-amplitude-damping}
\end{eqnarray}
\textchange{ parameterized by a real parameter $0 \le p \le 1$.}

\textchange{ We first calculate the key rate for the BB84 protocol.
In the BB84 protocol, we can estimate
the parameters $R_{\mathsf{zz}} = 1- p$, $R_{\mathsf{zx}} = 0$, 
$R_{\mathsf{xz}} = 0$, $R_{\mathsf{xx}} = \sqrt{1-p}$, 
$t_{\mathsf{z}} = p$, $t_{\mathsf{x}} = 0$. 
By the proposition \ref{proposition-non-imaginary}, we
can set $R_{\mathsf{zy}} = R_{\mathsf{xy}} = R_{\mathsf{yz}} = 
R_{\mathsf{yx}} = t_{\mathsf{y}} = 0$.
Furthermore, by the condition on the TPCP map \cite{fujiwara:99} }
\begin{eqnarray*}
(R_{\mathsf{xx}} - R_{\mathsf{yy}})^2 
\le (1 - R_{\mathsf{zz}})^2 - t_{\mathsf{z}}^2,
\end{eqnarray*}
\textchange{ we can decide the remaining parameter as $R_\san{yy} = \sqrt{1-p}$.
Thus, Eve's (worst-case) ambiguity $F(\omega)$ for the BB84 protocol coincide
with the true value $H_{\rho}(X|E)$, which means that the BB84 protocol can
achieve the same key rate as the six-state protocol. }

By straightforward calculations, the key rates of the 
direct reconciliation and reverse reconciliation are calculated as
\begin{eqnarray*}
1 + \frac{1}{2} h(p) - h\left( \frac{p}{2} \right)
- \frac{1+p}{2} h\left( \frac{1}{1+p} \right), 
\label{eq-key-rate-forward} 
\end{eqnarray*}
and
\begin{eqnarray*}
h\left( \frac{1+p}{2}\right) + 
\frac{1+p}{2}h\left(\frac{1}{1+p}\right) - h\left(\frac{1}{2}\right)
- \frac{1}{2}h(p)
\label{eq-key-rate-reverse}
\end{eqnarray*}
respectively. These key rates are
plotted in Fig.~\ref{Fig-comparison}.

\textchange{ The Bell diagonal entries of the Choi matrix
$(\rom{id} \otimes \mathcal{E}_p)(\psi)$ are
$\frac{1}{4}(2+2 \sqrt{1-p} - p)$,
$\frac{1}{4}p$, $\frac{1}{4}(2-2 \sqrt{1-p} - p)$,
and $\frac{1}{4}p$.
The key rate of the six-state protocol and the BB84
protocol with the conventional processing can be
calculated only from the Bell diagonal entries, and 
are also plotted in Fig.~\ref{Fig-comparison}. }

We find that the key rates of 
our proposed classical processing
are higher than those of the conventional processing. 
Furthermore, we find that the key rate of the reverse
reconciliation is higher than that of the direct
reconciliation.

\begin{figure}
\centering
\includegraphics[width=\linewidth]{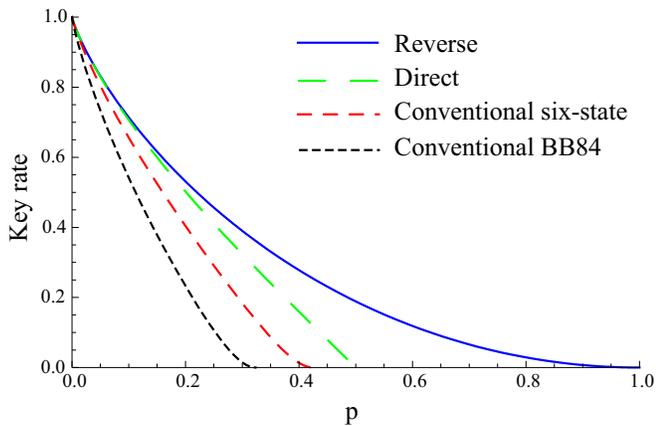}
\caption{ (Color online)
Comparison of the key rates against the parameter
$p$ of the amplitude damping channel (see Eq.~(\ref{eq-amplitude-damping})).
``Reverse'' and ``Direct'' are the key rates when we use the
reverse reconciliation and the direct reconciliation
in our proposed classical processing respectively.
``Conventional six-state'' and ``Conventional BB84'' are the key
rates of the six-state protocol and the BB84 protocol
with the conventional classical processing.
Note that the conventional classical processing 
involves the noisy preprocessing \cite{renner:05, kraus:05}.
}
\label{Fig-comparison}
\end{figure}

\begin{remark}
\label{remark-degradable}
When the channel is degradable \cite{devetak:05b}, i.e.,
there exists a channel $\mathcal{D}$ such that
$\mathcal{E}_E(\rho) = \mathcal{D} \circ \mathcal{E}_B(\rho)$
for any input $\rho$, 
the quantum wiretap channel capacity \cite{devetak:05} 
is known to be achievable without any auxiliary random variable
\cite{smith:07}. 

For the one-way key agreement from 
a degradable (from Alice to Bob and Eve) $\{ccq\}$-state, which is a state
$\rho_{XYE} = \sum_{x,y} P_{XY}(x,y) \ket{x}\bra{x} \otimes
 \ket{y}\bra{y} \otimes \rho_E^{x,y}$ such that there exist
states $\{ \hat{\rho}_E^{y} \}_y$ satisfying 
$\sum_{y} P_{Y|X}(y|x) \hat{\rho}_E^y = \rho_E^x := \sum_{y}
 P_{Y|X}(y|x) \rho_E^{x,y}$,
a similar statement also holds, namely
the key rate in Eq.~(\ref{eq-secure-key-rate}) cannot
be improved with any auxiliary random variable.
The use of auxiliary random variable for the key agreement
corresponds to the noisy preprocessing \cite{renner:05, kraus:05}.

The above statement is proved as follows.
Since we are considering the information reconciliation and 
the privacy amplification with one-way classical communication,
key rates only depend on distribution $P_{XY}$ and
$\{cq\}$-state $\rho_{XE}$. Thus the maximum key rate for
$\rho_{XYE}$ is equals to  that for degraded version of it,
$\hat{\rho}_{XYE} := \sum_{x,y} P_{XY}(x,y) \ket{x}\bra{x} \otimes
 \ket{y}\bra{y} \otimes \hat{\rho}_E^{y}$.
On the other hand the (quantum) intrinsic information
\begin{eqnarray*}
I_\rho(X;Y \downarrow E) := \inf I_\rho(X;Y | E^\prime)
\end{eqnarray*}
is an upper bound on the maximum key rate \cite{christandl:06}, where 
$I_\rho(X;Y | E^\prime) := H_\rho(XE) + H_\rho(YE) - H_\rho(XYE) - H_\rho(E)$ 
is the quantum conditional mutual information,
and the infimum is taken over all $\{ccq\}$-states
$\rho_{XYE^\prime} = (\rom{id} \otimes \mathcal{N}_{E \to
 E^\prime})(\rho_{XYE})$ for CPTP maps $\mathcal{N}_{E \to E^\prime}$
from system $E$ to $E^\prime$.
Taking the identity map $\rom{id}_E$, the quantum conditional mutual
information $I_\rho(X;Y|E)$ itself is an upper bound on
the maximum key rate.
Applying this fact for the degraded $\{ccq\}$-state, $\hat{\rho}_{XYE}$,
the maximum key rate is upper bounded by
\begin{eqnarray*}
\lefteqn{
I_{\hat{\rho}}(X;Y | E) 
} \\
&=& I_{\hat{\rho}}(X ; YE) - I_{\hat{\rho}}(X;E) \\
&=& H_{\hat{\rho}}(X|E) - H(X|Y) + I_{\hat{\rho}}(X;E | Y) \\
&=& H_{\rho}(X|E) - H(X|Y),
\end{eqnarray*}
which is the desired upper bound, and is equal
to Eq.~(\ref{eq-secure-key-rate}).

When Alice randomly sends $\{ \ket{0_\san{z}}, \ket{1_\san{z}} \}$ 
over the amplitude damping channel and Bob measures
the received state by $\sigma_\san{z}$ observable, 
the resulting $\{ccq\}$-state is 
degradable\footnote{The fact that amplitude damping channel is degradable
has been shown in \cite{giovannetti:05}.}, 
which implies
the key rate of direct reconciliation cannot be improved
by the noisy preprocessing.
It is not clear whether the $\{ccq\}$-state for the amplitude
damping channel is degradable in reverse order;  
there exists a possibility to improve the key rate of
reverse reconciliation by the noisy preprocessing.
\end{remark}

\subsection{Unital channel and rotation channel}
\label{example2}

\textchange{A channel $\mathcal{E}_B$ is called a unital channel if
the vector $(t_\san{z}, t_\san{x}, t_\san{y})$ is the zero
vector in the Stokes parameterization (see Eq.~(\ref{eq-affine-map})), 
or equivalently if the channel $\mathcal{E}_B$ maps
the completely mixed state $I/2$ to itself.
The unital channel has the following physical meaning
in QKD protocols. When Eve conducts the Pauli cloning \cite{cerf:00}
with respect to an orthonormal basis that is a rotated
version of $\{\ket{0_\san{z}}$, $\ket{1_\san{z}}\}$, the quantum channel
from Alice to Bob is not a Pauli channel but a unital
channel. It is natural to assume that Eve cannot determine
the direction of the basis $\{\ket{0_\san{z}}$, $\ket{1_\san{z}}\}$
accurately, and the unital channel deserve consideration
in the QKD research as well as the Pauli channel. 

By the singular value decomposition, we can 
decompose the matrix
$R$ in Eq.~(\ref{eq-affine-map}) as }
\begin{eqnarray}
\label{eq-svd}
O_2 \left[\begin{array}{ccc}
e_\san{z} & 0 & 0 \\
0 & e_\san{x} & 0 \\
0 & 0 & e_\san{y} 
\end{array} \right]
O_1,
\end{eqnarray}
\textchange{ where $O_1$ and $O_2$ are some rotation 
matrices\footnote{The rotation matrix is the real orthogonal
matrix with determinant $1$.}, and
$|e_\san{z}|$, $|e_\san{x}|$, and $|e_\san{y}|$
are the singular value of the matrix 
$R$\footnote{The decomposition is not unique 
because we can change the order of 
$(e_\san{z}, e_\san{x}, e_\san{y})$ or
the sign of them by adjusting the rotation matrices $O_1$ and $O_2$.
However, the result in this paper does not depends on a choice of the decomposition.}.  
Thus, we can consider the unital channel $\mathcal{E}_B$ as the composition of the
unitary channel $\mathcal{E}_{O_1}$, the Pauli channel }
\begin{eqnarray*}
\varrho \mapsto q_\san{i} \varrho + q_\san{z} \sigma_\san{z} \varrho \sigma_\san{z}
+ q_\san{x} \sigma_\san{x} \varrho \sigma_\san{x} + q_\san{y} \sigma_\san{y} \varrho \sigma_\san{y},
\end{eqnarray*}
\textchange{ and the unitary channel $\mathcal{E}_{O_2}$, where} 
\begin{eqnarray*}
q_\san{i} &=& \frac{1 + e_\san{z} + e_\san{x} + e_\san{y}}{4}, \\
q_\san{z} &=& \frac{1 + e_\san{z} - e_\san{x} - e_\san{y}}{4}, \\
q_\san{x} &=& \frac{1 - e_\san{z} + e_\san{x} - e_\san{y}}{4}, \\
q_\san{y} &=& \frac{1 - e_\san{z} - e_\san{x} + e_\san{y}}{4}
\label{eq-relation-q-d}
\end{eqnarray*}
\cite{bourdon:04}.

\textchange{
For the unital channel, we have 
$H(X|Y) = H(Y|X) = h((1+ R_\san{zz})/2)$.
For the six-state protocol, we can calculate Eve's 
ambiguity $H_\rho(X|E)$ as }
\begin{eqnarray}
1 - H[q_\san{i}, q_\san{z}, q_\san{x}, q_\san{y}] + 
h\left( \frac{1 + \sqrt{R_\san{zz}^2 + R_\san{xz}^2 + R_\san{yz}^2}}{2} \right)
\label{eq-unital-six-state}
\end{eqnarray}
\textchange{
because $(q_\san{i}, q_\san{z}, q_\san{x}, q_\san{y})$ are the eigenvalues
of the Choi matrix $\rho_{AB}$.
For the reverse reconciliation, Eve's ambiguity $H_\rho(Y|E)$ is given by
Eq.~(\ref{eq-unital-six-state}) in which $R_\san{xz}$ and $R_\san{yz}$
are replaced by $R_\san{zx}$ and $R_\san{zy}$ respectively.
Thus, $R_\san{xz}^2 + R_\san{yz}^2 = R_\san{zx}^2 + R_\san{zy}^2$
is the necessary and sufficient condition for $H_\rho(X|E) = H_\rho(Y|E)$.
For the BB84 protocol, we can calculate Eve's worst case ambiguity
$F(\omega)$ by Proposition \ref{remark-unital} because $t_\san{z} = t_\san{x} = 0$
for the unital channel. Note that the singular values $(d_\san{z}, d_\san{x})$
in Proposition \ref{remark-unital} are different from the singular values
$(|e_\san{z}|, |e_\san{x}|)$ in general because there exist
off-diagonal elements $(R_\san{zy}, R_\san{xy}, R_\san{yz}, R_\san{yx})$.
From Remark \ref{remark-after-unital}, 
$R_\san{xz}^2 = R_\san{zx}^2$ is the necessary and sufficient condition for that
Eve's worst case ambiguity for the direct reconciliation and that for the
reverse reconciliation coincide. }

\textchange{ In the rest of this section, we analyze a special class
of the unital channel, the rotation channel.
We define the rotation channel from Alice to Bob as }
\begin{eqnarray*}
\left[ \begin{array}{c}
\theta_{\mathsf{z}} \\ \theta_{\mathsf{x}} \\ \theta_{\mathsf{y}}
\end{array} \right]
\mapsto
\left[ \begin{array}{ccc}
\cos \vartheta & - \sin \vartheta & 0 \\
\sin \vartheta & \cos \vartheta & 0 \\
0 & 0 & 1
\end{array}
\right]
\left[ \begin{array}{c}
\theta_{\mathsf{z}} \\ \theta_{\mathsf{x}} \\ \theta_{\mathsf{y}}
\end{array} \right].
\end{eqnarray*}
\textchange{ The rotation channels occur, for example, when
the directions of 
the transmitter and the receiver
are not properly aligned. }

\textchange{For the rotation channel, Eq.~(\ref{eq-unital-bound})
gives $F(\omega) = 1$, which implies that Eve gained no information.
Thus, Eve's (worst-case) ambiguity for the BB84 protocol
coincide with the true value $H_{\rho}(X|E)$, and
the BB84 protocol with our proposed 
classical processing can achieve the same key rate
as the six-state protocol. }

There are two  reasons why we show 
this example---the rotation channel.
The first one is that we can obtain secret keys,
in the BB84 protocol, both from
matched measurement outcomes, which are transmitted and received 
by the same basis (say $\san{z}$-basis), and mismatched measurement outcomes,
which are transmitted and received by different bases
(say $\san{z}$-basis and $\san{x}$-basis respectively).
The probability distributions of
Alice and Bob's bit for each case are given by
$P_{X|Y}(1|0) = P_{X|Y}(0|1) = \sin^2 (\vartheta/2)$ and
$P_{X|Y^\prime}(1|0) = P_{X|Y^\prime}(0|1) = \sin^2 (\vartheta/2 -
\pi/4)$ respectively
(see Eqs.~(\ref{eq-distribution-xy}) and (\ref{eq-distribution-xy-2})
for the definitions of $P_{XY}$ and $P_{X Y^\prime}$).
If the channel is biased, i.e., $\vartheta \neq 0, \pi/2, \pi, 3 \pi /2$,
then we can obtain secret keys with positive key rates both
from matched measurement outcomes and mismatched measurement outcomes.
This fact solves an open problem discussed in \cite[Section 5]{matsumoto:07}.

\textchange{ The second reason is that we can obtain a secret key
from matched measurement outcomes even though the so called error rate is
higher than the $25$\% limit \cite{gottesman:03} in the BB84 protocol.
The Bell diagonal entries of
the Choi matrix $\rho_\vartheta$ 
are $\cos^2 (\vartheta/2)$, $0$, $0$,
and $\sin^2(\vartheta/2)$. Thus the error rate is $\sin^2(\vartheta/2)$.
For $\pi /3 \le \vartheta \le 5 \pi/3$,
the error rate is higher than $25$\%, but
we can obtain the positive key rate, $1 - h(\sin^2 (\vartheta/2))$
except $\vartheta = \pi/2, 3 \pi /2$.
Note that the key rate of the standard processing by
Shor and Preskill \cite{shor:00} is 
$1 - 2 h(\sin^2 (\vartheta/2))$.
This fact verifies Curty et al's suggestion \cite{curty:04} that
key agreement might be possible even for the error rates
higher than $25$\% limits.}

\begin{remark}
\label{remark-degraded}
If the $\{ccq\}$-state $\rho_{XYE}$ is degraded (from Alice to Bob and Eve), i.e.,
the $\{ccq\}$-state is of the form 
$\rho_{XYE} = \sum_{x,y} P_{XY}(x,y) \ket{x}\bra{x} \otimes
 \ket{y}\bra{y} \otimes \rho_E^y$, then 
we can prove that the key rate in Eq.~(\ref{eq-secure-key-rate}) cannot
be improved even if we use any noisy preprocessing or
two-way processing. 
The reason is that
the upper bound $I_\rho(X;Y|E)$ and the lower bound in 
Eq.~(\ref{eq-secure-key-rate}) coincide for the degraded $\{ccq\}$-state
in a similar manner to Remark \ref{remark-degradable}.

For the rotation channel $\mathcal{E}_\vartheta$, the resulting
$\{ccq\}$-state is obviously degraded. Thus the key rate
$1 - h(\sin^2 (\vartheta/2))$ cannot be improved any more. 
\end{remark}

\section{Conclusion}
\label{section-conclusion}

In this paper, we constructed a practically implementable 
classical processing for the BB84 protocol and the six-state protocol
that fully utilizes
the accurate channel estimation method.
A consequence of our result is that we should
not discard mismatched measurement outcomes in the QKD
protocols; those measurement outcomes can be used
to estimate the channel accurately, and increase
key rates. 

\textchange{ There is a problem that was not  treated in
this paper. Although we only treated asymptotically secure key 
rate in this paper, the final goal is the non-asymptotic analysis of 
eavesdropper's information, i.e.,
evaluate eavesdropper's information as a function of 
the length of the raw key, the key rate, and the length of sample bits
as in literatures \cite{mayers:01, biham:06, watanabe:06b-preprint,
renner:05b, hayashi:06, meyer:06, scarani:07, scarani:08}.
This topic is a future research agenda. }
 
\section*{Acknowledgment}

\textchange{ We would like to thank Dr.~Jean-Christian Boileau,
Professor Akio Fujiwara,
Dr.~Manabu Hagiwara, Dr.~Kentaro Imafuku, Professor Hideki Imai,
Professor Hoi-Kwong Lo and members of his group,
Professor Masahito Hayashi,
Dr.~Takayuki Miyadera,
Professor Hiroshi Nagaoka,
Professor Renato Renner and members of his group,
Mr.~Yutaka Shikano,
Professor Tadashi Wadayama,
Professor Stefan Wolf and members of his group,
and Professor Isao Yamada
for valuable discussions and comments. }

\textchange{We also would like to appreciate the first 
reviewer for letting us know the reference \cite{liang:03}
and the second reviewer for pointing out an error in 
Remark \ref{remark-asymmetric}
of the earlier version of the manuscript.}

This research was partly supported by the Japan
Society for the Promotion of Science under
Grants-in-Aid No.~18760266 and
No.~00197137.

\appendix


\section{Convexity of Eve's ambiguity}
\label{appendix-a-}

In this appendix, we show a lemma that will be used in the rest
of the appendices.

\begin{lemma}
\label{proposition-convexity}
For two channels $\mathcal{E}_B^1$ and 
$\mathcal{E}_B^2$, and a probabilistically mixed channel 
$\mathcal{E}_B^\prime := \lambda \mathcal{E}_B^1 + (1-\lambda) \mathcal{E}_B^2$, 
Eve's ambiguity is convex, i.e., we have
\begin{eqnarray*}
H_{\rho^\prime}(X|E) \le \lambda H_{\rho^1}(X|E)
 + (1-\lambda) H_{\rho^2}(X|E),
\end{eqnarray*}
where $\rho^\prime_{X E} := \sum_{x \in \mathbb{F}_2} \frac{1}{2}
 \ket{x}\bra{x} \otimes \mathcal{E}_E^\prime(
 \ket{x}\bra{x})$ for channel
$\mathcal{E}_E^\prime$ to all the environment of $\mathcal{E}_B^\prime$,
and $\rho_{XE}^r := \sum_{x \in \mathbb{F}_2} \frac{1}{2}
 \ket{x}\bra{x} \otimes \mathcal{E}_E^r(
 \ket{x}\bra{x})$ for channel $\mathcal{E}_E^r$ to
all the environment of $\mathcal{E}_B^r$ and for $r \in \{1,2\}$.
\end{lemma}
\begin{proof}
For $r = 1$ and $2$, let $\psi_{ABE}^r$ be a purification of the 
Choi matrix $\rho_{AB}^r := (\rom{id} \otimes \mathcal{E}_B^r)(\psi)$.
Then the density matrix $\rho_{XE}^r$ is derived by
Alice's measurement by $\san{z}$-basis and the partial trace over Bob's system,
 i.e.,
\begin{eqnarray}
\label{eq-cq-state-1}
\rho_{XE}^r = \rom{Tr}_B \left[
\sum_{x} (\ket{x}\bra{x} \otimes I) \psi^r_{ABE} 
   (\ket{x}\bra{x} \otimes I)
\right].
\end{eqnarray}
Let 
\begin{eqnarray*}
\ket{\psi^\prime_{ABER}} := 
   \sqrt{\lambda} \ket{\psi_{ABE}^1} \ket{1} + 
   \sqrt{1-\lambda} \ket{\psi_{ABE}^2} \ket{2} 
\end{eqnarray*}
be a purification of $\rho^\prime_{AB} := (\rom{id} \otimes \mathcal{E}_B^\prime)(\psi)$, where $\mathcal{H}_R$ is the reference system,
and $\{ \ket{1}, \ket{2} \}$ is an orthonormal basis of $\mathcal{H}_R$.
Let 
\begin{eqnarray}
\label{eq-cq-state-2}
\rho^\prime_{XER} :=
    \rom{Tr}_B \left[
\sum_{x} (\ket{x}\bra{x} \otimes I) \psi^\prime_{ABER} 
   (\ket{x}\bra{x} \otimes I)
\right],
\end{eqnarray}
and let 
\begin{eqnarray*}
\rho^{*}_{XER} &:=& \sum_{r \in \{1,2\}} 
   (I \otimes \ket{r}\bra{r}) \rho^\prime_{XER} 
   (I \otimes \ket{r}\bra{r}) \\
&=& \lambda \rho_{XE}^1 \otimes \ket{1}\bra{1}
    + (1 - \lambda) \rho_{XE}^2 \otimes \ket{2}\bra{2}
\end{eqnarray*}
be the density matrix such that  the system $\mathcal{H}_R$ is measured
 by $\{ \ket{1}, \ket{2}\}$ basis.
Then we have
\begin{eqnarray*}
\lefteqn{
H_{\rho^\prime}(X|ER) 
} \\
&=& H(X) - I_{\rho^\prime}(X; ER) \\
&\le& H(X) - I_{\rho^*}(X;ER) \\
&=& H_{\rho^*}(X|ER) \\
&=& \lambda H_{\rho^1}(X|E) + (1-\lambda) H_{\rho^2}(X|E),
\end{eqnarray*}
where the inequality follows from the monotonicity
of the quantum mutual information for measurements
(data processing inequality) \cite{hayashi-book:06}.
By renaming the systems $ER$ to $E$, we have the 
assertion of the lemma.
\end{proof}

\begin{remark}
\label{remark-extention}
By switching the role of Alice and Bob, we can show that the assertion in
Lemma \ref{proposition-convexity} and 
thus Proposition \ref{proposition-non-imaginary} hold for 
the reverse reconciliation.
Furthermore, the statements also hold for the information reconciliation
and the privacy amplification
with $k$-block-wise two-way processing \cite{gottesman:03, watanabe:07} 
(including one-way noisy  preprocessing \cite{kraus:05, renner:05}).
More precisely, let $\mathcal{N}_{X^k Y^k \to UV}$ be the
TPCP map that represents a two-way processing.
Then for the density matrix  
\begin{eqnarray*}
\rho_{UVE^k} := (\mathcal{N}_{X^k Y^k \to UV} \otimes \rom{id}_{E^k})
(\rho_{XYE}^{\otimes k}),
\end{eqnarray*}
we can obtain the inequality 
\begin{eqnarray*}
H_{\rho^\prime}(U|VE^k) \le \lambda H_{\rho^1}(U|VE^k)
  + (1- \lambda) H_{\rho^2}(U|VE^k).
\end{eqnarray*}
The modifications of the proof is to replace $\psi_{ABE}^r$ and
$\psi_{ABER}^\prime$ with $(\psi_{ABE}^r)^{\otimes k}$ and
$(\psi_{ABER}^\prime)^{\otimes k}$ in Eqs.~(\ref{eq-cq-state-1})
and (\ref{eq-cq-state-2}), to replace
the partial trace over Bob's system with Bob's measurement,
to append a map $\mathcal{N}_{X^k Y^k \to UV}$, and to replace
the measurement on the system $\mathcal{H}_R$ with 
the measurements on $\mathcal{H}_R^{\otimes k}$.
\end{remark}

\section{Proof of Proposition \ref{proposition-non-imaginary}}
\label{appendix-a}

The statement of the Proposition \ref{proposition-non-imaginary}
easily follows from Lemma \ref{proposition-convexity}.
For any channel $\mathcal{E}_B$,
let $\bar{\mathcal{E}}_B$ be the channel whose Choi matrix is
the complex conjugate of that for $\mathcal{E}_B$.
Note that eigenvalues of density matrices are unchanged
by the complex conjugate, and thus Eve's ambiguity
$H_{\bar{\rho}}(X|E)$ for $\bar{\mathcal{E}}_B$ equals to 
$H_\rho(X|E)$.
By applying Lemma \ref{proposition-convexity}
for $\mathcal{E}_B^1 = \mathcal{E}_B$, $\mathcal{E}_B^2 = \bar{\mathcal{E}}_B$,
and $\lambda = \frac{1}{2}$, 
we have 
\begin{eqnarray*}
H_{\rho^\prime} \le \frac{1}{2} H_{\rho}(X|E) + 
    \frac{1}{2} H_{\bar{\rho}}(X|E),
\end{eqnarray*}
where $\rho^\prime_{AB} = \frac{1}{2} \rho_{AB} + \frac{1}{2}
\bar{\rho}_{AB}$.
\textchange{ Note that $\rho^\prime_{AB}$ is
a real density matrix whose entries are equal to
the real components of $\rho_{AB}$, which implies   
that the parameters $R_{\san{zy}}$, $R_{\san{xy}}$, $R_{\san{yz}}$, $R_{\san{yx}}$,
and $t_{\san{y}}$, are $0$ by Eq.~(\ref{eq-choi-matrix}).
Since the channel $\mathcal{E}_B$ was arbitrary, we have
the assertion of the proposition. 
\qed }


\section{Proof of Lemma \ref{lemma-continuity}}
\label{appendix-b}

Since the conditional entropy is a continuous function,
the following statement is suffice for proving that
$F(\omega)$ is continuous function at any $\omega_0 \in \mathcal{P}$,
where $\mathcal{P}$ is the set of all $\omega$ such that
$\mathcal{P}(\omega)$ is not empty.
For any $\omega \in \mathcal{P}$ such that $\| \omega - \omega_0 \| \le
 \varepsilon$,
there exist $\varepsilon^\prime, \varepsilon^{\prime\prime} >0$ such that
\begin{eqnarray}
\label{eq-neighbor-1}
\mathcal{P}(\omega) &\subset& \mathcal{B}_{\varepsilon^\prime}(\mathcal{P}(\omega_0)), \\
\label{eq-neighbor-2}
\mathcal{P}(\omega_0) &\subset& \mathcal{B}_{\varepsilon^{\prime\prime}}(\mathcal{P}(\omega)),
\end{eqnarray}
and $\varepsilon^\prime$ and $\varepsilon^{\prime\prime}$ converge to
$0$ as $\varepsilon$ goes to $0$,
where $\mathcal{B}_{\varepsilon^\prime}(\mathcal{P}(\omega_0))$ is the
$\varepsilon^\prime$-neighbor of the set $\mathcal{P}(\omega_0)$.

Define the set  $\mathcal{Q} := \{(\omega, R_\san{yy}) \mid \omega \in \mathcal{P}, R_\san{yy}
 \in \mathcal{P}(\omega) \}$, which is a closed convex set.
Define functions 
\begin{eqnarray*}
U(\omega) &:=& \max_{R_\san{yy} \in \mathcal{P}(\omega)} R_\san{yy}, \\
L(\omega) &:=& \min_{R_\san{yy} \in \mathcal{P}(\omega)} R_\san{yy}
\end{eqnarray*}
as the upper surface and the lower surface of the set $\mathcal{Q}$ respectively.
Then $U(\omega)$ and $L(\omega)$ are concave and convex functions respectively, 
because $\mathcal{Q}$ is a convex set.
Thus $U(\omega)$ and $L(\omega)$ are continuous functions except
the extreme points of $\mathcal{P}$.
For any extreme point $\omega^\prime$ and for any interior point
 $\omega$,
we have $U(\omega) \ge U(\omega^\prime)$ and 
$L(\omega) \le L(\omega^\prime)$, because
$\mathcal{Q}$ is a convex set.
Since $\mathcal{Q}$ is a closed set, we have
$\lim_{\omega \to \omega^\prime} U(\omega) \in \mathcal{P}(\omega^\prime)$
and $\lim_{\omega \to \omega^\prime} L(\omega) \in \mathcal{P}(\omega^\prime)$,
which implies that $U(\omega^\prime) = \lim_{\omega \to \omega^\prime}
 U(\omega)$ 
and $L(\omega^\prime) = \lim_{\omega \to \omega^\prime}
 L(\omega)$.
Thus $U(\omega)$ and $L(\omega)$ are also continuous at the extreme
 points.
Since $\mathcal{P}(\omega)$ is a convex set, the continuity
of $U(\omega)$ and $L(\omega)$ implies that
Eqs.~(\ref{eq-neighbor-1}) and (\ref{eq-neighbor-2})
hold for some $\varepsilon^\prime, \varepsilon^{\prime\prime} >0$,
and $\varepsilon^\prime$ and $\varepsilon^{\prime\prime}$ converge
to $0$ as $\varepsilon$ goes to $0$.
\qed

\section{Proof of Proposition \ref{remark-unital}}
\label{appendix-d}

\textchange{ By Proposition \ref{proposition-non-imaginary},
it suffice to consider the channel $\mathcal{E}_B$ of the form }
\begin{eqnarray*}
\left[ \begin{array}{c}
\theta_{\mathsf{z}} \\ \theta_{\mathsf{x}} \\ \theta_{\mathsf{y}}
\end{array} \right]
\mapsto
\left[ \begin{array}{ccc}
R_{\mathsf{z}\mathsf{z}} & R_{\mathsf{z}\mathsf{x}} & 0 \\
R_{\mathsf{x}\mathsf{z}} & R_{\mathsf{x}\mathsf{x}} & 0 \\
0 & 0 & R_{\mathsf{y}\mathsf{y}}
\end{array}
\right]
\left[ \begin{array}{c}
\theta_{\mathsf{z}} \\ \theta_{\mathsf{x}} \\ \theta_{\mathsf{y}}
\end{array} \right]
+  
\left[ \begin{array}{c}
t_{\mathsf{z}} \\ t_{\mathsf{x}} \\ 0
\end{array} \right].
\end{eqnarray*}
\textchange{ Define the channel $\mathcal{E}_B^-(\varrho) := \sigma_\san{y} ( \mathcal{E}_B(\sigma_\san{y} \varrho \sigma_\san{y}))\sigma_\san{y}$ and the mixed channel $\mathcal{E}_B^\prime := \frac{1}{2} \mathcal{E}_B + \frac{1}{2} \mathcal{E}_B^-$.
Since the channel $\mathcal{E}_B^-$ is given by }
\begin{eqnarray*}
\left[ \begin{array}{c}
\theta_{\mathsf{z}} \\ \theta_{\mathsf{x}} \\ \theta_{\mathsf{y}}
\end{array} \right]
\mapsto
\left[ \begin{array}{ccc}
R_{\mathsf{z}\mathsf{z}} & R_{\mathsf{z}\mathsf{x}} & 0 \\
R_{\mathsf{x}\mathsf{z}} & R_{\mathsf{x}\mathsf{x}} & 0 \\
0 & 0 & R_{\mathsf{y}\mathsf{y}}
\end{array}
\right]
\left[ \begin{array}{c}
\theta_{\mathsf{z}} \\ \theta_{\mathsf{x}} \\ \theta_{\mathsf{y}}
\end{array} \right]
+  
\left[ \begin{array}{c}
- t_{\mathsf{z}} \\ - t_{\mathsf{x}} \\ 0
\end{array} \right],
\end{eqnarray*}
$\mathcal{E}_B^\prime$ is a unital channel and the matrix part of $\mathcal{E}_B$
and $\mathcal{E}_B^\prime$ are the same. Furthermore, since $H_\rho(X|E)$ for
$\mathcal{E}_B$ equals to $H_{\rho^-}(X|E)$ for $\mathcal{E}_B^-$,
by using Lemma \ref{proposition-convexity}, we have 
\begin{eqnarray*}
H_\rho(X|E) \ge H_{\rho^\prime}(X|E).
\end{eqnarray*}

\textchange{ The rest of the proof is to calculate the minimization of $H_{\rho^\prime}(X|E)$
with respect to $R_\san{yy}$. By the singular value decomposition, we can decompose
the matrix $R^\prime$ corresponding to the channel $\mathcal{E}^\prime_B$ as}
\begin{eqnarray*}
O_2 \left[ \begin{array}{ccc}
\tilde{d}_\san{z} & 0 & 0 \\
0 & \tilde{d}_\san{x} & 0 \\
0 & 0 & R_\san{yy} 
\end{array} \right] O_1,
\end{eqnarray*}
\textchange{ where $O_1$ and $O_2$ are some rotation matrices within $\san{z}-\san{x}$-plane, and 
$|\tilde{d}_\san{z}|$ and $|\tilde{d}_\san{x}|$ are the 
singular value of the matrix in Eq.~(\ref{eq-matrix}).
Then, we have}
\begin{eqnarray*}
\lefteqn{ \min_{R_\san{yy}} H_{\rho^\prime}(X|E) } \\
&=& \min_{R_\san{yy}} \left[
1 - H(\rho_{AB}^\prime) + \sum_{x \in \mathbb{F}_2} \frac{1}{2} 
H(\mathcal{E}_B^\prime(\ket{x}\bra{x})) \right] \\
&=& 1 - \max_{R_\san{yy}} H[q_\san{i}, q_\san{z}, q_\san{x}, q_\san{y}] 
 + h\left(\frac{1 + \sqrt{R_\san{zz}^2 + R_\san{xz}^2}}{2} \right) \\
&=& 1 - h(q_\san{i} + q_\san{z}) - h(q_\san{i} + q_\san{x}) 
 + h\left(\frac{1 + \sqrt{R_\san{zz}^2 + R_\san{xz}^2}}{2} \right),
\end{eqnarray*}
\textchange{ where $(q_\san{i}, q_\san{z}, q_\san{x}, q_\san{y})$ are
the eigenvalues of the Choi matrix $\rho_{AB}^\prime$.
By noting that $q_\san{i} + q_\san{z} = \frac{1 + \tilde{d}_\san{z}}{2}$
and $q_\san{i} + q_\san{x} = \frac{1 + \tilde{d}_\san{x}}{2}$
(see Section \ref{example2}), we
have assertion of the proposition. \qed } 

\section{Convex Optimization}
\label{appendix-convex-opt}

In this appendix, we briefly explain how to apply a convex optimization
method, the interior-point method, to the channel estimation in
the BB84 protocol. In a similar manner, we can apply the interior-point method 
to the channel estimation in the six-state protocol.
For more detail, see the textbook 
\cite[Section 11.6]{boyd-book:04}.

First, we define a generalized inequality.
Since the set $K \subset \mathbb{R}^{4 \times 4}$ of (real) semi-definite matrices is 
a proper cone (see \cite[Section 2.4.1]{boyd-book:04} for the definition of the 
proper cone), we can define a generalized inequality $\preceq_K$ as
\begin{eqnarray*}
M \preceq_K N \Longleftrightarrow N - M \in K.
\end{eqnarray*}

For a given parameter $(\omega, R_\san{yy}) \in \mathbb{R}^7$, we define
the real matrix $\rho(\omega,R_\san{yy}) \in \mathbb{R}^{4 \times 4}$ by
using the relation in Eq.~(\ref{eq-choi-matrix}),
where we set other parameters 
$(R_{\san{zy}}, R_{\san{xy}}, R_{\san{yz}}, R_{\san{yx}}, t_{\san{y}})$ to be
all $0$.  Then, 
the function $\rho: \mathbb{R}^7 \to \mathbb{R}^{4 \times 4}$
is a $K$-concave function (see \cite[Section 3.6.2]{boyd-book:04} for the definition
of the $K$-concave function).

We can formulate our optimization problem as follows:
\begin{eqnarray*}
\begin{array}{ll}
\mbox{minimize} & \| \hat{\omega} - \tilde{\omega} \|^2 \\
\mbox{subject to} & \rho(\hat{\omega},R_\san{yy}) \succeq_K 0, \\
 & \rom{Tr}_B[ \rho(\hat{\omega},R_\san{yy}) ] = I,
\end{array}
\end{eqnarray*}
where $\| \cdot \|^2$ is the square Euclidean norm, which is
a convex function, and $I$ is the $2 \times 2$ identity matrix.
This optimization problem can be solved by the interior-point
method. Note that we can use $\log \det \rho(\hat{\omega},R_\san{yy})$ as a logarithmic 
barrier function (see \cite[Example 11.7]{boyd-book:04}). 


\end{document}